\documentclass[journal,10pt,final]{IEEEtran}

\usepackage{cite}
\usepackage{graphicx}
\usepackage{algorithm}
\usepackage{algorithmic}
\usepackage{epsfig}
\usepackage{epstopdf}
\usepackage{amssymb,amsmath}
\usepackage{amsmath}
\usepackage{amsthm}
\usepackage{amsfonts}
\usepackage{subfigure}
\usepackage{psfrag}
\usepackage{rotating}
\usepackage{latexsym}
\usepackage{dsfont}
\usepackage{stmaryrd}
\usepackage{amsthm}
\usepackage{amssymb}
\usepackage{amsmath}

\usepackage[force,almostfull]{textcomp}
\usepackage{fancyvrb}
\usepackage{textcomp}
\usepackage{url}
\usepackage{ifdraft}
\usepackage{url}
\usepackage{multirow}
\usepackage{rotating}
\usepackage{xspace}
\usepackage{array}
\usepackage{arydshln}
\usepackage{multido}
\usepackage{flushend}
\usepackage{enumerate}
\usepackage[latin1]{inputenc}
\newlength\figureheight 
\newlength\figurewidth 
\usepackage{graphics}
\usepackage{pgfplots}
\pgfplotsset{compat=newest}
\pgfplotsset{plot coordinates/math parser=false}

\usepackage{scalefnt}

\newtheorem{specialcasecounter}{Theorem}
\newtheoremstyle{specialcasestyle}{1mm}{1mm}{\upshape}{}{\bfseries\upshape}{.}{0mm}{}
\theoremstyle{specialcasestyle}

\newtheorem{assump}{Assumption}
\newtheorem{lem}{Lemma}
\newtheorem{rem}{Remark}

\begin{document}
\title{A Fast Simulation Method for the Sum of Subexponential Distributions$^*$}

\author{Nadhir~Ben Rached, Fatma~Benkhelifa, Abla Kammoun, Mohamed-Slim~Alouini, and Raul Tempone\\\thanks{ \hspace{-3mm}The authors are members of the KAUST Strategic Research Initiative on Uncertainty Quantification in Science and Engineering (SRI-UQ).}
\small Computer, Electrical and Mathematical Science and Engineering (CEMSE) Division\\
\small King Abdullah University of Science and Technology (KAUST)\\
\small Thuwal, Makkah Province, Saudi Arabia \\
\small \{nadhir.benrached, fatma.benkhelifa, abla.kammoun, slim.alouini, raul.tempone\}@kaust.edu.sa
}
\date{}
\maketitle
\thispagestyle{empty}

\begin{abstract}
Estimating the probability that a sum of random variables (RVs) exceeds a given threshold is a well-known challenging problem. Closed-form expression of the sum distribution is usually intractable and presents an open problem. A crude Monte Carlo (MC)  simulation is the standard technique for the estimation of this type of probability. However, this approach is computationally expensive especially when dealing with rare events (i.e events with very small probabilities). Importance Sampling (IS) is an alternative approach which effectively improves the computational efficiency of the MC simulation. In this paper, we develop a general framework based on IS approach for the efficient estimation of the probability that the sum of independent and not necessarily identically distributed heavy-tailed RVs exceeds a given threshold. The proposed IS approach is based on constructing a new sampling distribution by twisting the hazard rate of the original underlying distribution of each component in the summation. A minmax approach is carried out for the determination of the twisting parameter, for any given threshold. Moreover, using this minmax optimal choice, the estimation of the probability of interest is shown to be asymptotically optimal as the threshold goes to infinity. We also offer some selected simulation results illustrating first the efficiency of the proposed IS approach compared to the naive MC simulation. The near-optimality of the minmax approach is then numerically analyzed.
\end{abstract}


\section{Introduction}

{ \ The performance analysis of communication systems is generally associated with the investigation of the statistics of sums of Random Variables (RVs). For instance, when diversity techniques such as Maximum ratio Combining (MRC) and Equal Gain combining (EGC) are performed, the resulting received signal-to-noise-ratio (SNR) is modeled by a sum of fading variates \cite{alouini}. 
\par }
{\ Unfortunately, the statistics of the sum distribution for most of the challenging problems are generally intractable and unknown. Monte Carlo (MC) simulation is the standard technique to estimate the probability that a sum of RVs exceeds a given threshold. However, this approach requires an extensive computational work to estimate extremely small probabilities. Importance Sampling (IS) is an alternative approach which aims to improve the computational efficiency of the naive MC simulation technique \cite{rubino2009rare}. The basic idea behind IS technique is to change the underlying sampling distribution in a way to achieve a substantial variance reduction of the IS estimator. Many research efforts have been carried out to propose efficient IS algorithms. For instance, among the first works in the digital communication field, the authors in \cite{192731} and \cite{31142} proposed methods based respectively on scaling the variance and shifting the mean of the original probability measure. An extension of \cite{192731} was performed in \cite{52645} where a composite IS technique was derived. In \cite{212308}, the asymptotic efficiency of five different IS techniques was studied for the estimation of the Bit Error Rate (BER) in digital communication systems with Gaussian input. Exponential twisting, derived from the large deviation theory, is an interesting IS change of measure technique since in most of the cases it yields "optimal" asymptotic results \cite{54903} \cite{179349}. For instance, this technique was used to estimate the BER of direct-detection optical systems employing avalanche photodiode (APD) receivers in \cite{380042}. 
\par }
{\ The exponential twisting change of measure is feasible only with distributions having finite Moment Generating Function (MGF). Thus, in the heavy-tailed setting where the MGF is infinite, it is not possible to use the exponential twisting method. However, many heavy-tailed distributions, such as the Log-normal and the Weibull (with shape parameter less than 1) RVs, are frequently encountered in various applications. In cellular mobile communication systems, the Co-Channel Interference (CCI) power which arises due for instance to the neighboring cells that use the same frequency is generally modeled as a sum of Log-normal (SLN) RVs \cite{Stuber:2001:PMC:368633}. Besides, the Log-normal distribution is also used to model the large-scale fading in the ultra-wideband (UWB) communications \cite{Ghavami}, and the weak-to-moderate turbulence channels in free-space optical communication channels \cite{journals/twc/NavidpourUK07}. Recently, the Weibull fading has also received an increasing attention since it exhibits a good fit to experimental fading data for both indoor and outdoor environment\cite{1512431}, \cite{837048}, \cite{900150}. Various closed-form approximations of the sum of Log-normal RVs \cite{citeulike:6297231} \cite{citeulike:7151841} \cite{1275712} \cite{1369233} and the sum of Weibull RVs \cite{1665128} \cite{1388722} \cite{alouini1} have been extensively developed. These approximations are not generic and depend on the problem under consideration. Hence, a lot of research efforts have focused in developing a generic efficient IS technique dealing with distributions in the heavy-tailed class. In \cite{asmussen2000}, two efficient techniques for the estimation of the probability that the sum of subexponential RVs exceeds a given threshold have been presented. The first one is based on conditional MC, whereas the second is based on considering a new probability measure which is heavier than the underlying distribution. In \cite{TLR}, a transform likelihood ratio approach was derived to switch the heavy-tailed problem into an equivalent light-tailed one. The authors in \cite{Juneja:2002:SHT:566392.566394} have developed an efficient fast simulation method for estimating a sum of independent and identically distributed (i.i.d) RVs with subexponential decay. Their approach is based on twisting the hazard rate of the original probability measure of each component in the summation.
\par }
In this paper, inspired by \cite{Juneja:2002:SHT:566392.566394}, we develop a general approach based on hazard rate twisting to efficiently estimate the probability that a sum of independent and non-identically distributed heavy-tailed RVs exceeds a given threshold. The twisting parameter is determined through a minmax approach which first ensures a nearly optimal computational gain in terms of the number of simulation runs and second leads to an asymptotic optimality criterion. The rest of the paper is organized as follows. In section II, we state the problem setting and enumerate the main contributions. In Section III, a minmax hazard rate twisting approach is introduced with an emphasis on the general procedure leading to an efficient choice of the twisting parameter. Moreover, the asymptotic optimality criterion using this proposed IS approach is verified. In Section IV, two applications of distributions with subexponential decay are studied. In Section V, a substantial computational gain of the proposed IS technique is analyzed and shown through various selected simulation results.

\section{Mathematical Background}
\subsection{Problem Setting}
{\ Let $X_1,X_2,...,X_N$ be a sequence of independent but not necessarily identically distributed positive RVs. Let us denote the Probability Density Function (PDF) of each $X_i$ by $f_i(x)$, $i=1,2,...,N$. Our objective is to efficiently estimate 
\begin{align}
\alpha=\mathbb{P}\left (\sum_{i=1}^{N}{X_i}>\gamma_{th}\right )=P \left (S_N>\gamma_{th} \right),
\end{align}
for a sufficiently large threshold $\gamma_{th}$. We focus on heavy-tailed distributions, i.e distributions which exhibit slower decays than the exponential distribution. Formally, a distribution of a RV $X$ is said to be heavy-tailed if
\begin{align}
\lim_{x \rightarrow +\infty}{\exp \left(\nu x \right )\mathbb{P}\left(X>x\right)}=+\infty, \text{  for all  } \nu >0.
\end{align}
In practice, all commonly used heavy-tailed distributions belong to the subclass of subexponential distributions. In fact, a distribution of a RV $X$ is said to be subexponential if 
\begin{align}
\overline{F^{*n}}(x) \sim n\overline{F}(x) \hspace{2mm}\text{as} \hspace{2mm} x \rightarrow +\infty,
\end{align}
where $\overline{F}(x)$ is the Complementary Cumulative Distribution Function (CCDF) of $X$, and $\overline{F^{*n}}(x)$ is the CCDF of the sum of $n$ i.i.d RVs with distribution $F$. Examples of such subexponential distributions are: the Log-normal distribution, and the Weibull distribution with shape parameter less than $1$. The readers are referred to \cite{Juneja:2002:SHT:566392.566394} for more discussion about subexponential distributions.\par}
{\ The standard technique to estimate $\alpha$ is to use the naive MC estimator defined as 
\begin{align}
\hat \alpha_{MC}= \frac{1}{M}\sum_{j=1}^{M}{\textbf{1}_{\left (S_N(\omega_j)>\gamma_{th}\right )}},
\end{align}
where $M$ is the number of simulation runs, and $\textbf{1}_{(\cdot)}$ defines the indicator function. It is widely known that the naive MC simulation is extensively expensive for the estimation of rare events. In fact, from the Central Limit Theorem (CLT), it can be shown that the MC estimation with $10\%$ relative precision requires more than $100/\alpha$ simulation runs. Hence the number of samples to estimate a probability of order $10^{-9}$ should be more than $10^{11}$, with an accuracy requirement of $10\%$. Consequently, there is a detrimental need to improve the computational efficiency of the MC simulation.\par}
\subsection{Importance Sampling}
{\ IS is a variance reduction technique which aims to increase the computational efficiency of the naive MC simulation \cite{rubino2009rare}. The general concept of IS is to construct an unbiased estimator of the desired probability with much smaller variance than the naive estimator. In fact, this technique is based on performing a suitable change of the sampling distribution as follows
\begin{align}
\nonumber\alpha&=\int_{\mathbb{R}^N}{\textbf{1}_{(S_N>\gamma_{th})}f_1(x_1)f_2(x_2)...f_N(x_N)}\\
\nonumber &=\int_{\mathbb{R}^N}{\textbf{1}_{(S_N>\gamma_{th})}L\left(x_1,x_2,...,x_N\right)g_1(x_1)g_2(x_2)...g_N(x_N)}\\
&=\mathbb{E}_{p^*}\left[\textbf{1}_{(S_N>\gamma_{th})}L\left(X_1,X_2,...,X_N\right)\right],
\end{align}
where the expectation is taken with respect to the new probability measure $p^*$ under which the PDF of each $X_i$ is $g_i$, and $L$ is the likelihood ratio defined as
\begin{align}\label{lik}
L\left(X_1,X_2,...,X_N\right)=\prod_{i=1}^{N}{\frac{f_i(X_i)}{g_i(X_i)}}.
\end{align}
The idea behind this change of measure is to enhance sampling important points which have more impact on the desired probability. Hence, emphasizing that important points are sampled frequently will result in a decrease of the variance of the IS estimator. The new IS estimator is defined as
\begin{align}
\hat \alpha_{IS}=\frac{1}{M}\sum_{i=1}^{M}{\textbf{1}_{(S_N(\omega_i)>\gamma_{th})}L(X_1(\omega_i),...,X_N(\omega_i))}.
\end{align}\par}
{\ \\Generally, it is not obvious how to construct a new probability measure which results in decreasing the variance of the IS estimator and hence increasing the computational efficiency. Besides, it is necessary to define some performance metrics which measure the goodness and the pertinence of the IS estimator. Bounded relative error, asymptotic optimality, and bounded likelihood ratio are useful indicators to characterize a good change of probability measure \cite{rubino2009rare}. Generally, it is difficult to achieve the bounded relative error criterion, whereas the asymptotic optimality could be shown if one choose an appropriate probability measure $g_i$. Let us consider the sequence of the RVs $\{ T_{\gamma_{th}}\}$ defined as 
\begin{align}
T_{\gamma_{th}}=\textbf{1}_{(S_N>\gamma_{th})}L\left(X_1,...,X_N\right).
\end{align}
From the non-negativity of the variance of $T_{\gamma_{th}}$, we get
\begin{align}
\mathbb{E}_{p^*} \left [T_{\gamma_{th}}^2 \right ] \geq (\mathbb{P}(S_N>\gamma_{th}))^2.
\end{align}
Applying the logarithm on both side, we conclude that for all $p^*$ we have
\begin{align}
{\frac{\log\left(\mathbb{E}_{p^*}\left[T^2_{\gamma_{th}}\right]\right)}{\log\left(\mathbb{P}\left(S_N>\gamma_{th}\right)\right)}} \leq 2.
\end{align}
Hence, we say that $\alpha$ is asymptotically optimally estimated under the probability measure $p^*$ if the above equation holds with equality as $\gamma_{th} \rightarrow +\infty$, that is
\begin{align}\label{asymp_opt}
\lim_{\gamma_{th} \rightarrow \infty} {\frac{\log\left(\mathbb{E}_{p^*}\left[T^2_{\gamma_{th}}\right]\right)}{\log\left(\mathbb{P}\left(S_N>\gamma_{th}\right)\right)}}=2.
\end{align}
It is important to note that the naive simulation is not asymptotically optimal for the estimation of $\alpha$ since the ratio in (\ref{asymp_opt}) is equal to $1$.\par}
The exponential twisting technique, which is derived from the large deviation theory, is the main IS framework dealing with light-tailed distributions, that is distributions whose tails decay at an exponential rate or faster. The exponential twisting by an amount $\theta >0$ is given by 
\begin{align}\label{expo}
g_i\left ( x\right )\triangleq f_{i,\theta}(x)=\frac{f_i(x)\exp(\theta x)}{M_{X_i}(\theta)},
\end{align}
where $M_{X_i}(\theta)$ denotes the moment generating function (MGF) of the RV $X_i$. In most of the cases, this technique achieves the asymptotic optimality criterion given in (\ref{asymp_opt}) \cite{179349}.  

In the heavy-tailed setting, the exponential twisting change of measure is not feasible and alternative techniques are needed. In fact, the MGFs are infinite for distributions with heavy tails. In \cite{Juneja:2002:SHT:566392.566394}, an efficient IS technique was developed for the estimation of $\alpha$ in the case of i.i.d sum of RVs with subexponential decay. Their idea was based on twisting the hazard rate of each component in the summation $S_N$ by a quantity $0 <\theta <1$.
{\ Let us define the hazard rate $\lambda_i(\cdot)$ associated to the RV $X_i$ as
\begin{align}\label{hazard rate}
\lambda_i(x)=\frac{f_i(x)}{1-F_i(x)}, \hspace{2mm} x>0,
\end{align}
where $F_i(\cdot)$ is the CDF of $X_i$ , $i=1,...,N$. Besides, we define also the hazard function as
\begin{align}\label{hazard function}
\nonumber \Lambda_i(x)&=\int_{0}^{x}{\lambda_i(t)dt}\\
&=-\log \left (1-F_i(x) \right ), \hspace{2mm}x>0.
\end{align}
From (\ref{hazard rate}) and (\ref{hazard function}), the PDF of $X_i$ is related to the hazard rate and function as
\begin{align}
\nonumber f_i(x)&=\lambda_i(x)\exp\left(-\int_{0}^{x}{\lambda_i(t)dt}\right)\\
&=\lambda_i(x)\exp\left(-\Lambda_i(x)\right).
\end{align}
The change of probability measure is obtained by twisting the hazard rate of the underlying distribution by a quantity $0<\theta<1$ as follows
\begin{align}\label{twisted pdf}
\nonumber g_i(x)& \triangleq f_{i,\theta}(x)=\left(1-\theta\right)\lambda_i(x)\exp\left(-\left(1-\theta\right)\Lambda_i(x)\right)\\
&=\left(1-\theta\right)f_i(x)\exp\left(\theta\Lambda_i\left(x\right)\right).
\end{align}
Consequently, the RV $T_{\gamma_{th}}$ has the following expression
\begin{align}\label{tgamma}
T_{\gamma_{th}}=\frac{1}{\left(1-\theta\right)^N}\exp\left(-\theta\sum_{i=1}^{N}{\Lambda_i(X_i)}\right)\textbf{1}_{(S_N>\gamma_{th})}.
\end{align}
\par}
{\ For heavy-tailed distributions, the hazard rate twisting based approach plays the same role as the exponential twisting technique in the light-tailed setting. In \cite{1371507}, the authors emphasize the central role played by hazard rate technique for the estimation of small probabilities that a general function containing both light and heavy-tailed distributions exceeds a given threshold. In fact, by developing log-asymptotic expressions for both the probability of interest and the second moment of $T_{\gamma_{th}}$, they have proved that $\alpha$ is asymptotically  optimally estimated. The equivalence between the hazard rate and the exponential twisting techniques is also emphasized in \cite{1261434} where a suitable hazard function transformation is used, in the case of a sum of i.i.d subexponential distributions, to switch from a heavy-tailed problem to a light-tailed one where the exponential twisting could be used.
\par}
\subsection{Main Contributions}
{\ A primordial question when using either exponential or hazard rate twisting techniques is the choice of the twisting parameter $\theta$. The selection of this parameter should be performed in a way to ensure a maximum reduction of the second moment of $T_{\gamma_{th}}$, and hence result in a maximum amount of computational gain. Unfortunately, this is not feasible in general since $\mathbb{E}_{\theta}\left [ T_{\gamma_{th}}^2 \right ]$ ( $\mathbb{E}_{\theta} \left [ \cdot\right ]$ denotes the expectation under the IS probability measure ) is typically not known in a closed form. The commonly used procedure to determine $\theta$ starts by deriving a close upper bound on $\mathbb{E}_{\theta}\left [ T_{\gamma_{th}}^2 \right ]$ and then finding the value of $\theta$ which minimizes that upper bound. For the exponential twisting, this upper bound  is easily obtained using (\ref{expo}) and (\ref{lik})
\begin{align}
\nonumber \mathbb{E}_{\theta} \left [L^2 \textbf{1}_{(S_N>\gamma_{th})} \right ]&=\mathbb{E}_{\theta}\left [M_{S_N}^2(\theta)\exp \left(-2\theta S_N \right )\textbf{1}_{(S_N>\gamma_{th})}\right]\\
& \leq M_{S_N}^2(\theta)\exp \left (-2\theta \gamma_{th} \right).
\end{align}
Then, the value of $\theta=\theta^*$ selected to minimize the upper bound is satisfying
\begin{align}
\frac{M_{S_N}^{'}(\theta^*)}{M_{S_N}(\theta^*)}=\gamma_{th}.
\end{align}
\par}
{\ In the hazard rate twisting setting, the determination of $\theta^*$ is not as straightforward as for the exponential twisting approach. In fact, the upper bound on the second moment is not easy to obtain. In \cite{Juneja:2002:SHT:566392.566394}, the i.i.d sum of subexponential distributions is considered. The determination of the twisting parameter was done via the derivation of an upper bound on the second moment of $T_{\gamma_{th}}$ which holds only for a sufficiently large threshold. More precisely, by assuming that the hazard rates are eventually decreasing to zero and are eventually everywhere differentiable, the asymptotic inequality
\begin{align}\label{asymp_iid}
\sum_{i=1}^{N}{\Lambda(x_i)} \geq \Lambda \left (\sum_{i=1}^{N}{x_i} \right )-\epsilon ,
\end{align}
holds for every $\epsilon >0$ and with $\sum_{i}^{N}{x_i}$ large enough. Then, using the previous asymptotic inequality, an upper bound on $\mathbb{E}_{\theta} \left [ T_{\gamma_{th}}^2 \right ]$ was computed which is minimized when 
\begin{align}\label{thetaiid}
\theta=1-\frac{N}{\Lambda(\gamma_{th})}.
\end{align}
Moreover, they proved in \cite{Juneja:2002:SHT:566392.566394} that asymptotic optimality holds by replacing $N$ in (\ref{thetaiid}) by any positive constant. In the present work, we consider a non-trivial generalization of \cite{Juneja:2002:SHT:566392.566394} to the case of the sum of independent and non-identically distributed subexponential RVs. Our procedure for the determination of the twisting parameter is performed in two steps. First, we derive an upper bound on the second moment of $T_{\gamma_{th}}$ through the resolution of a constrained maximization problem on the likelihood ratio. Second, we minimize this upper bound over all possible value of $\theta$ which results in the so called minmax optimal twisting parameter $\theta=\theta^*$. For the class of subexponential distributions, we will see that, under a weaker assumption than the one stated in \cite{Juneja:2002:SHT:566392.566394} to derive (\ref{asymp_iid}), we are able to characterize the behavior of the solution of the maximization problem and detect the region where the maximum is achieved.
\par}
{\ In a nutshell, the main contributions of the present paper are:
\begin{itemize}
\item We develop an optimized hazard rate twisting approach for the estimation of $\alpha$ for the case of the sum of independent and non-identically distributed subexponential RVs. The procedure that we will follow to determine $\theta$ is based on a minmax approach. This minmax procedure starts by computing the maximum (the most sharpest upper bound) on the second moment of $T_{\gamma_{th}}$ for all value of $\gamma_{th}$. Then, a simple minimization problem is solved to derive the minmax optimal twisting parameter $\theta^*$. Besides, we will see also that this choice of $\theta$ is  efficient since it almost results in the same computational gain as the unknown optimal value (the value that minimizes the actual second moment of $T_{\gamma_{th}}$). In the particular i.i.d sum, we prove that our minmax twisting parameter is equivalent to the one derived in \cite{Juneja:2002:SHT:566392.566394} as $\gamma_{th}$ goes to infinity.
\item We prove under some realistic assumptions, which are generally satisfied by distributions with subexponential decays, that $\alpha$ is asymptotically optimally estimated using our minmax approach. 
\item Finally, two applications will be studied to clarify how the procedure is applied, and to validate through numerical results the efficiency of the proposed minmax hazard rate twisting approach. The first application considers the sum of independent Log-normal RVs, and the second one deals with the sum of independent Weibull distributions with shape parameter less than $1$. It is important to note that in our approach there is no restriction to consider the sum of a mixture of subexponential distributions belonging to different families.
\end{itemize}
\par}
\section{Proposed Hazard Rate Twisting}
\subsection{General Approach}
{\ Generally, an interesting IS change of probability measure for the estimation of rare events is characterized by the property of uniformly bounded likelihood ratio. This property will result in obtaining an upper bound on the second moment of the RV $T_{\gamma_{th}}$. Then, the optimal value of the parameter $\theta$ is chosen to minimize that upper bound.
More precisely, the procedure of choosing $\theta$ is divided into two steps. In the first step, we construct an upper bound on the second moment of $E_{\theta}(T_{\gamma_{th}}^2)$ which is achieved by solving the following maximization problem (P):
\begin{align}
\nonumber (P): \underset{X_1,...,X_N}{\max} \hspace{3mm}&L(X_1,X_2,...,X_N)\\
\text{Subject to   } &\sum_{i=1}^{N}{X_i} \geq\gamma_{th},\\
\nonumber &X_i>0, \hspace{2mm} i=1,...,N,
\end{align}
where the likelihood ratio is given as follows
\begin{align}
L(X_1,X_2,...,X_N)=\frac{1}{\left(1-\theta\right)^N}\exp\left(-\theta\sum_{i=1}^{N}{\Lambda_i(X_i)}\right).
\end{align}
Hence, solving the problem $(P)$ is equivalent to solving the following minimization problem $(P')$:
\begin{align}
\nonumber (P'): \underset{X_1,...,X_N}{\min} \hspace{3mm}& \sum_{i=1}^{N}{\Lambda_i(X_i)}\\
\text{Subject to   } &\sum_{i=1}^{N}{X_i} \geq\gamma_{th},\\
\nonumber &X_i>0, \hspace{2mm} i=1,...,N.
\end{align}
The resolution of the maximization problem $(P)$ or equivalently the minimization problem $(P')$ will be discussed later in the following subsection. 
\par }
{\ By denoting the optimal solution of $(P)$ by $X_1^*,X_2^*,...,X_N^*$, we have 
\begin{align}\label{bound}
\nonumber \mathbb{E}_{\theta}\left[T_{\gamma_{th}}^2\right]&=\mathbb{E}_{\theta}\left[L^2\left(X_1,X_2,...,X_N\right)\textbf{1}_{(S_N>\gamma_{th})}\right]\\
&\leq \frac{1}{\left(1-\theta\right)^{2N}}\exp\left(-2\theta\sum_{i=1}^{N}{\Lambda_i(X_i^*)}\right).
\end{align}
The second step is to minimize (\ref{bound}) to get the optimal twisting parameter $\theta^*$. This is a simple minimization problem to solve which results in
\begin{align}\label{theta}
\theta^*&=1-\frac{N}{\sum_{i=1}^{N}{\Lambda_i(X_i^*)}}.
\end{align} 
\par }
\subsection{Asymptotic Optimality Criterion}
{\ The value of the twisting given in (\ref{theta}) represents the minmax optimal choice among all values of $\theta$, and for all threshold values. Now, we focus on the asymptotic behavior of the IS estimator as $\gamma_{th}$ goes to infinity. In particular, we investigate the asymptotic optimality criterion (\ref{asymp_opt}) using the twisting parameter $\theta^*$ given in (\ref{theta}). 
\par }
{\ The investigation of the asymptotic optimality criterion is based on analyzing the asymptotic behavior of the solution of the minimization problem $(P')$. Since each hazard function $\Lambda_i(\cdot)$ is an increasing function, it follows that the inequality constraint is satisfied with equality, that is
\begin{align}\label{cons}
\sum_{i=1}^{N}{X_i^*}=\gamma_{th}.
\end{align} 
In order to ensure the asymptotic optimality, let us consider the following assumption
\begin{assump}
\hspace{2mm} For each $i \in \{1,2,...,N\}$, we assume that there exist $\eta_i$ such that the hazard function $\Lambda_i(\cdot)$ is concave in the interval $[\eta_i, +\infty )$.
\end{assump} 
The previous assumption is consistent with all commonly used subexponential distributions such as the Log-normal, the Weibull (with shape parameter less than 1), and the Pareto (with parameter strictly bigger than 1) distributions. In the following lemma, we characterize the behavior of the solution of $(P')$ for a sufficiently large threshold $\gamma_{th}$:
\begin{lem} \hspace{2mm}Under Assumption 1, there exists a fixed index $i_0 \in \{ 1,2,...,N \}$ such that the minimizers of $(P')$ satisfy for a sufficiently large $\gamma_{th}$
\begin{align}
\gamma_{th}- &\sum_{i\neq i_0}{\eta_i} \leq X_{i_0}^*\leq \gamma_{th},\\
X_i &\leq \eta_i, \text{  for all  } i \neq i_0,
\end{align}
and hence as $\gamma_{th}\rightarrow +\infty$, we have  
\begin{align}
X_{i_0}^* &\underset{+\infty}{\sim} \gamma_{th},\text{  as  } \gamma_{th} \rightarrow \infty,\\
X_i^*&=\mathcal{O}(1), \text{  for all  }i\neq i_0.
\end{align}
\end{lem}
\begin{proof}
Let us consider $S(N,\gamma_{th})$ the set of all feasible solutions: 
\begin{align}
S(N,\gamma_{th})=\{   X=(X_1,X_2,...,X_N) \in (\mathbb{R}^+)^N, \sum_{i=1}^{N}{X_i}=\gamma_{th} \}.
\end{align}
Through the use of Assumption 1, the objective function of $(P')$ is concave on the subset:
\begin{align}
\nonumber \tilde{S}(N,\gamma_{th})&=\{ X=(X_1,X_2,...,X_N) \in (\mathbb{R}^+)^N, \sum_{i=1}^{N}{X_i}=\gamma_{th},\\
&X_i  \geq\eta_i, \text{ for each }i\in \{1,2,...,N \} \}.
\end{align}
Thus, the minimum of the objective function of $(P')$ over $\tilde{S}(N,\gamma_{th})$ is achieved in at least one of its extreme points. More precisely, the extreme points of $\tilde{S}(N,\gamma_{th})$ are $e_1,e_2,...,e_N$ such that $e_i=(\eta_1,\eta_2,...,\eta_{i-1},\gamma_{th}-\sum_{j\neq i}{\eta_j},\eta_{i+1},...,\eta_N)$. Therefore the minimum of $(P')$ over $S(N,\gamma_{th})$ is either achieved in one of the extreme point $e_i$, $i=1,2,...,N$, or on the set
\begin{align}
\nonumber &\bar {S}(N,\gamma_{th})=S(N,\gamma_{th})\backslash \tilde{S}(N,\gamma_{th})\\
\nonumber &=\{X=(X_1,X_2,...,X_N)\in (\mathbb{R}^+)^N,\sum_{i=1}^{N}{X_i}=\gamma_{th},\\
 &\exists i \text{ such that } X_i <  \eta_i  \}.
\end{align}
In both cases, there exists at least one index $i \in \{ 1,2...,N\}$ such that $X_i^*\leq \eta_i$. In addition, in order to satisfy the equality constraint $\sum_{i=1}^{N}{X_i^*}=\gamma_{th}$ for a sufficiently large $\gamma_{th}$, there should exist an index $j \in \{1,2,...,N \}$ such that $X_j \geq \eta_j$. In order to prove the result in Lemma 1, we proceed iteratively by dimension reduction. In fact, without loss of generality, we assume that $X_N^*\leq \eta_N$ (through an index permutation). It follows that
\begin{align}
\min_{S(N,\gamma_{th})}{\sum_{i=1}^{N}{\Lambda_i(X_i)}}=\min_{X_{N} \leq \eta_N} \min_{S(N-1,\gamma_{th,N-1})}{\sum_{i=1}^N{\Lambda_i(X_i)}},
\end{align}
where $\gamma_{th,N-1}=\gamma_{th}-X_{N}$, it follows that
\begin{align}
\min_{S(N,\gamma_{th})}{\sum_{i=1}^{N}{\Lambda_i(X_i)}}=\Lambda_N(X_N^*)+\min_{S(N-1,\gamma_{th,N-1}^*)}{\sum_{i=1}^{N-1}{\Lambda_i(X_i)}},
\end{align}
Consequently, we can see that we have reduced the number of optimization variables to be $N-1$, while we have kept the same structure of the minimization problem $(P')$ with $\gamma_{th,N-1}^*=\gamma_{th}-X_{N}^*$. Hence the previous procedure could be repeated again. In fact, using the same argument as before, there exists another index $i\in\{1,2,...,N-1 \}$ such that $X_i^* \leq  \eta_i$. Without loss of generality, we assume that $i=N-1$ which leads to
\begin{align}
\nonumber \min_{S(N,\gamma_{th})}{\sum_{i=1}^{N}{\Lambda_i(X_i)}}&=\Lambda_N(X_N^*)+\Lambda_{N-1}(X_{N-1}^*)\\
&+\min_{S(N-2,\gamma_{th,N-2}^*)}{\sum_{i=1}^{N-2}{\Lambda_i(X_i)}},
\end{align}
where $\gamma_{th,N-2}^*=\gamma_{th}-X_{N}^*-X_{N-1}^*$. After $N-2$ steps, we get
\begin{align}
\nonumber \min_{S(N,\gamma_{th})}{\sum_{i=1}^{N}{\Lambda_i(X_i)}}&=\sum_{i=1}^{N-2}{\Lambda_{N+1-i}(X_{N+1-i}^*)}\\
&+\min_{S(2,\gamma_{th,2}^*)}{\sum_{i=1}^{2}{\Lambda_i(X_i)}},
\end{align}
with $X_i^* \leq \eta_i$, for $i=3,4....,N$, and $\gamma_{th,2}=\gamma_{th}-\sum_{i=3}^{N}{X_i^*}$. Thus, we end up with a two dimensional minimization problem. Again, there should exist an index $i=2$ ( through a possible permutation ) such that $X_2^* \leq \eta_2$. Therefore, using the equality constraint $\sum_{i=1}^N{X_i^*}=\gamma_{th}$, we get
\begin{align}
X_i^*&\leq \eta_i, \text{  }i=2,3,...,N,\\
\gamma_{th,2}^*&-\eta_2\leq X_1^* \leq \gamma_{th,2}^*. 
\end{align}
Since $\eta_i$, $i=2,3,...,N$ are independent of $\gamma_{th}$, it follows
\begin{align}
\gamma_{th}-\sum_{i=2}^{N}{\eta_i} \leq X_1^*\leq \gamma_{th}.
\end{align}
Thus, as $\gamma_{th}$ goes to infinity, we have
\begin{align}
X_1^* &\underset{+\infty}{\sim} \gamma_{th}\\
X_i^*&=\mathcal{O}(1), \text{ }\forall i \in \{2,3,...,N  \}.
\end{align} 
\end{proof}
It is important to note that in the particular i.i.d case, the index $i_0$ could be any index in $\{1,2,...,N \}$, and the minimum is achieved in $N$ different points. A direct consequence of Lemma 1 is presented in the following lemma.
\begin{lem} \hspace{2mm} Under Assumption 1, the objective function of $(P')$ has the following asymptotic behavior
\begin{align}\label{consequence}
\sum_{i=1}^{N}{\Lambda_i(X_i^*)} \underset{+\infty}{\sim} \Lambda_{i_0}(\gamma_{th}), \text{  as  }\gamma_{th} \rightarrow +\infty.
\end{align}
\end{lem}
\begin{proof}
Using Lemma 1 and the fact that $\Lambda_{i_0}(\gamma_{th})$ tends to infinity as $\gamma_{th}$ increases, we have 
\begin{align}
\frac{\Lambda_i(X_i^*)}{\Lambda_{i_0}(\gamma_{th})} \rightarrow 0 \text{  as  }\gamma_{th} \rightarrow +\infty , \text{ for all }i\neq i_0.
\end{align}
The remaining work is to prove that 
\begin{align}\label{eq2}
\frac{\Lambda_{i_0}(X_{i_0}^*)}{\Lambda_{i_0}(\gamma_{th})} \underset{+\infty}{\sim} 1, \text{  as  }\gamma_{th} \rightarrow +\infty.
\end{align}
Using the fact that $\Lambda_{i_0}(\cdot)$ is increasing to infinity and concave for inputs bigger than $\eta_{i_0}$, then its derivative which is the hazard rate $\lambda_{i_0}(\cdot)$ is a decreasing function provided that $x \geq \eta_{i_0}$. Hence, $\lambda_{i_0}(\cdot)$ is bounded by $\lambda_{i_0}(\eta_{i_0})$ for all $x \geq \eta_{i_0}$. Consequently, $\Lambda_{i_0}(\cdot)$ is Lipschitz in the interval $[\eta_{i_0},+\infty)$ and we have
\begin{align}
\Lambda_{i_0}(\gamma_{th})-\Lambda_{i_0}(X_{i_0}^*)=\mathcal{O}(\gamma_{th}-X_{i_0}^*), \text{  as  }\gamma_{th} \rightarrow +\infty.
\end{align}
Using Lemma 1, we have that $\gamma_{th}-X_{i_0}^*=\mathcal{O}(1)$. Thus, it follows that
\begin{align}
\Lambda_{i_0}(\gamma_{th})-\Lambda_{i_0}(X_{i_0}^*)=o(\Lambda_{i_0}(\gamma_{th})),
\end{align}
which leads to (\ref{eq2}) and then the proof is concluded.
\end{proof}
Now, we could state the asymptotic optimality theorem
\begin{specialcasecounter}
For a sum of independent RVs with subexponential distributions and under Assumption 1, the quantity of interest $\alpha$ is asymptotically optimally estimated using the hazard rate twisting approach with the minmax optimal twisting parameter $\theta^*$ given in (\ref{theta}).
\end{specialcasecounter}
\par }
\begin{proof}
In (\ref{bound}), we have derived an upper bound on the second moment of $T_{\gamma_{th}}$ as
\begin{align}
\mathbb{E}_{\theta^*}\left[T_{\gamma_{th}}^2\right] \leq \frac{1}{\left(1-\theta^*\right)^{2N}}\exp\left(-2\theta^*\sum_{i=1}^{N}{\Lambda_i(X_i^*)}\right).
\end{align}
By setting $A(\gamma_{th})=\sum_{i=1}^{N}{\Lambda_i(X_i^*)}$ and replacing the optimal twisting parameter $\theta^*$ given in (\ref{theta}), we have
\begin{align}
\mathbb{E}_{\theta^*}\left[T_{\gamma_{th}}^2\right] \leq \left ( \frac{A(\gamma_{th})}{N}\right )^{2N} \exp \left( -2A(\gamma_{th})+2N\right).
\end{align}
By applying the logarithmic function on both side, it follows
\begin{align}\label{log}
\log \left ( \mathbb{E}_{\theta^*}\left[T_{\gamma_{th}}^2\right]\right ) \leq 2N \left ( 1+\log(\frac{A(\gamma_{th})}{N}) \right )-2A \left (\gamma_{th} \right).
\end{align}
On the other hand, using the non-negativity of $X_i, i\in \{1,2,...,N\}$, we have 
\begin{align}\label{asym}
\log\left(\alpha\right)=\log\left(P \left (  S_N > \gamma_{th}\right )\right) \geq \log\left(P \left (  X_{i_0} >\gamma_{th}\right )\right).
\end{align}
Note that for sufficiently large $\gamma_{th}$, the left and right-hand sides of \eqref{log} are negative. Therefore,
\begin{align}\label{final_res}
\frac{\log\left(\mathbb{E}_{\theta^*}\left[T_{\gamma_{th}}^2\right]\right)}{\log(\alpha)}\geq \frac{2N \left ( 1+\log(\frac{A(\gamma_{th})}{N}) \right )-2A(\gamma_{th})}{-\Lambda_{i_0}(\gamma_{th})}.
\end{align}
Finally, using Lemma 2, we have:
\begin{align}
\frac{2N \left ( 1+\log(\frac{A(\gamma_{th})}{N}) \right )-2A(\gamma_{th})}{-\Lambda_{i_0}(\gamma_{th})} &\underset{+\infty}{\sim}\frac{-2A(\gamma_{th})}{-\Lambda_{i_0}(\gamma_{th})}\nonumber \\
&\underset{+\infty}{\sim} 2 \label{52}.
\end{align}
Through the use of the non-negativity of the variance, we conclude the proof.
\end{proof}
\begin{rem}
\hspace{2mm} Under the i.i.d case, Assumption 1 is almost equivalent to the one stated in \cite{Juneja:2002:SHT:566392.566394}. They assumed also that the hazard rate is converging to zero whereas in our case this is not needed. The previous observation makes our assumption a bit weaker compared to \cite{Juneja:2002:SHT:566392.566394}. In addition, our optimal twisting parameter $\theta^*$ given in (\ref{theta}) tends to the same value (\ref{thetaiid}) derived in \cite{Juneja:2002:SHT:566392.566394}, as $\gamma_{th}$ goes to infinity. 
\end{rem}
\subsection{Generation of the Twisted Distribution}
{\ Generally, hazard rate twisting the original PDF of a RV $X$ does not result in a known distribution. One way to generate realizations of $X$ under $f_{\theta}(\cdot)$ could be performed via its CDF $F_{\theta}(\cdot)$. In fact, it is known that $F^{-1}_{\theta}(U)$, where $U$ is uniformly distributed RV over $[0,1]$, has the same distribution as $X$ under the hazard rate twisted PDF \cite{devroye:1986}. Let us consider a RV $X$ with an underlying PDF $f(\cdot)$ and CDF $F(.)$. From (\ref{twisted pdf}), the PDF $f_{\theta}(\cdot)$ associated to $X$ with hazard rate $\lambda(\cdot)$ and hazard function $\Lambda(\cdot)$ is
\begin{align}
\nonumber f_{\theta}(x)&=(1-\theta)\lambda(x) \exp(-(1-\theta)\Lambda(x))\\
&=(1-\theta)f(x)\exp(\theta \Lambda(x)).
\end{align}
Replacing $\lambda(\cdot)$ and $\Lambda(\cdot)$ by their definitions, we get
\begin{align}
f_{\theta}(x)=\frac{(1-\theta)f(x)}{(1-F(x))^{\theta}}.
\end{align}
By a simple integration, the corresponding CDF is given by
\begin{align}
F_{\theta}(x)=-\frac{1}{(1-F(x))^{\theta -1}}+1.
\end{align}
Finally, a simple computation leads to an exact expression of the CDF inverse of the RV $X$ under the hazard rate twisting technique  
\begin{align}\label{cdf_inv}
F_{\theta}^{-1}(y)=F^{-1}(1-(1-y)^{-\frac{1}{\theta -1}}),
\end{align}
where $F^{-1}(\cdot)$ is the CDF inverse of $X$ under the original PDF $f(\cdot)$.  A pseudo-code describing all steps to estimate $\alpha$ by our proposed hazard rate twisting approach is described in Algorithm 1.
\begin{algorithm}[H]
\caption{Optimized hazard rate twisting approach for the estimation of $\alpha$}
\begin{algorithmic}
\STATE \textbf{Inputs:} $M_{IS}$, $\gamma_{th}$.
\STATE \textbf{Outputs:} $\hat \alpha_{IS}$.
\STATE Find the optimal value of $\theta$ as in $(\ref{theta})$ by solving the minimization problem $(P')$.
\FOR {$i=1,...,M_{IS}$}
\STATE Generate $N$ independent realizations of the uniform distribution over $[0,1]$:  $U_1(\omega_i),U_2(\omega_i),...,U_N(\omega_i)$.
\STATE Compute $X_1(\omega_i),X_2(\omega_i),...,X_N(\omega_i)$ using $(\ref{cdf_inv})$ : $X_j(\omega_i)=F_{\theta}^{-1}(U_j(\omega(i))),j=1,2,...,N$
\STATE Evaluate $T_{\gamma_{th}}(\omega_i)$  as in  $(\ref{tgamma})$.
\ENDFOR
\STATE Compute the IS estimator as $\hat \alpha_{IS}=\frac{1}{M_{IS}}\sum_{i=1}^{M_{IS}}{T_{\gamma_{th}}(\omega_i)}$.
\end{algorithmic}
\end{algorithm}
\par }
\section{Applications}
{\ We consider two examples of distributions belonging to the class of subexponential distributions: the Log-normal and the Weibull (with shape parameter less than $1$) distributions. We will investigate for these two examples the solution of $(P')$.
\subsection{Weibull Distribution}
In this example, the PDF of $X_i$, $i=1,2,...,N$ is 
\begin{align}
f_i(x)=\frac{k_i}{\beta_i}\left (\frac{x}{\beta_i}\right )^{k_i-1} \exp \left (- \left (\frac{x}{\beta_i}\right )^{k_i}\right ), \hspace{2mm} x\geq 0.
\end{align}
where $k_i>0$ and $\beta_i>0$ denotes respectively the shape and the scale parameters. We focus on the case where the shape parameter is strictly less than $1$ since it is known that with this choice the Weibull RV is a subexponential distribution. The hazard rate and function for each $X_i$ are as follow
\begin{align}
\lambda_i(x)&=\frac{k_i}{\beta_i}\left (\frac{x}{\beta_i}\right )^{k_i-1},\\
\Lambda_i(x)&=\left (\frac{x}{\beta_i}\right )^{k_i}.
\end{align}
Let us now investigate the solution of the minimization problem $(P')$. We could prove through a simple computation that the objective function of $(P')$ is concave for $k_i<1$, $i=1,2,...,N$ and hence Assumption 1 is satisfied. In fact, the Hessian $H$ of the objective function at any point $X=(X_1,X_2,...,X_N)\in (\mathbb{R}^{+})^{N}$ is a diagonal matrix with diagonal elements
\begin{align}
[H(X_1,X_2,...,X_N)]_{ii}=\frac{k_i(k_i-1)}{\beta_i^2}\left (\frac{X_i}{\beta_i}\right )^{k_i-2},
\end{align}
which is strictly negative for $k_i<1$, $i=1,2,...,N$. In particular, the objective function is also concave on the convex set $S(N,\gamma_{th})=\{ X=(X_1,X_2,...,X_N) \in (\mathbb{R}^{+})^{N}, \text{ such that } \sum_{i=1}^{N}{X_i}=\gamma_{th}\}$. Therefore, the solution of $(P')$ is obtained in one of the extreme points of $S(N,\gamma_{th})$. In other words, the minimum is achieved when
\begin{align}
X_{i_0}^*=\gamma_{th}, \text{ and } X_i^*=0 \text{  } \forall i\neq i_0,
\end{align} 
where $i_0$ satisfying 
\begin{align}
\left (\frac{\gamma_{th}}{\beta_{i_0}}\right )^{k_{i_0}} \leq \left (\frac{\gamma_{th}}{\beta_{i}}\right )^{k_{i}}, \text{  } \forall i\neq i_0.
\end{align}
It is important to note that for large values of $\gamma_{th}$, the index $i_0$ depends only on the shape and scale parameters and independent of $\gamma_{th}$. More precisely, for $\gamma_{th}$ large enough, it is characterized by
\begin{align}
i_0=\operatorname{arg\,min}_i k_i.
\end{align}
Moreover, if there are more than one RV with minimum shape parameter, the index $i_0$ corresponds to the one with maximum scale parameter. 
\begin{rem}
\hspace{2mm} We have described in the previous section a method based on the inverse CDF $F_{\theta}^{-1}(\cdot)$ to generate samples of a RV $X$ under the twisted PDF $f_{\theta}(\cdot)$. For the particular Weibull distribution with parameters $k$ and $\beta$, the PDF $f_{\theta}(\cdot)$ is simply another Weibull distribution with the same shape parameter $k$ and a different scale parameter $\beta'$ as follows
\begin{align}
\nonumber f_{\theta}(x)&=(1-\theta)\lambda(x)\exp\left (-(1-\theta)\Lambda(x)\right )\\
\nonumber &=(1-\theta)\frac{k}{\beta}\left (\frac{x}{\beta}\right )^{k-1}\exp\left (-(1-\theta)(\frac{x}{\beta})^k\right )\\
&=\frac{k}{\beta'}\left (\frac{x}{\beta'}\right )^{k-1}\exp \left (-(\frac{x}{\beta'})^k\right ).
\end{align} 
where $\beta'=\frac{\beta}{(1-\theta)^{1/k}}$.
\end{rem}
\subsection{Log-Normal Distribution}
The PDF of each $X_i$, $i=1,2,..,N$ is given by 
\begin{align}
f_i(x)=\frac{1}{\sqrt{2\pi}\sigma_ix}\exp\left(-\frac{\left(\log(x)-\mu_i\right)^2}{2\sigma_i^2}\right), \hspace{2mm}x>0,\label{pdf}
\end{align}
where $\mu_i$ and $\sigma_i$ are the mean and the standard deviation of the associated Gaussian RV $Y_i=\log(X_i)$. In communication, the decibel unit is generally used. Hence, it is more convenient to define a Gaussian RV as $Z_i=10\log_{10}(X_i)$ with mean $\mu_{i,dB}$ and standard deviation $\sigma_{i,dB}$. The relation between the two Gaussian RVs $Y_i$ and $Z_i$, $i=1,2,...,N$, are
\begin{align}
 \mu_{i}=\xi \mu_{i,dB} \text{ and }\sigma_{i}=\xi \sigma_{i,dB}
\end{align}
where $\xi=\log(10)/10$. The expressions of $\lambda_i(\cdot)$ and $\Lambda_i(\cdot)$ are given by
\begin{align}
\lambda_i(x)&=\frac{\frac{1}{x\sigma_i} \phi\left(\frac{\log(x)-\mu_i}{\sigma_i} \right) }{1-\Phi\left(\frac{\log(x)-\mu_i}{\sigma_i} \right) },\\
\Lambda_i(x)&=-\log\left(1-\Phi\left(\left(\frac{\log(x)-\mu_i}{\sigma_i} \right) \right)\right),
\end{align}
where $\phi(\cdot)$ and $\Phi(\cdot)$ are respectively the PDF and the CDF of a standard Gaussian distribution. In this example, the solution of $(P')$ is not straightforwardly computed as the Weibull distribution. The difficulty to find out the analytic solution of the minimization problem $(P')$ arises from the fact that the hazard function for a Log-normal RV does not have the concavity property as for the Weibull distribution. However, it is known that the hazard function of the Log-normal distribution has the property stated in Assumption 1. Therefore, the minimizers of $(P')$ satisfies Lemma 1 which states that there exists an index $i_0$ such that $X_{i_0}^*$ is close to $\gamma_{th}$ whereas the other components are bounded. Hence, since the hazard function $\Lambda_i$ is an increasing function, the index $i_0$ satisfies for a sufficiently large $\gamma_{th}$
\begin{align}
\left(\log(\gamma_{th})-\mu_{i_0}\right)/\sigma_{i_0} \leq \left(\log(\gamma_{th})-\mu_{i}\right)/\sigma_{i}, \forall i \neq i_0.
\end{align}
Thus, for $\gamma_{th}$ large enough, the index $i_0$ is independent of $\gamma_{th}$ and corresponds to 
\begin{align}
i_0=\operatorname{arg\,max}\sigma_i.
\end{align}
Moreover, if there exists another index with a maximum standard deviation, $i_0$ corresponds to the RV with a maximum mean. 
\section{Simulation Results}
{\ In this section, some selected simulation results are shown to compare the naive MC simulation and the proposed IS simulation technique. Two performance metrics will be used to compare these two approaches. The relative error of the naive MC estimator is defined through the use of the CLT as
\begin{align}\label{MC}
\epsilon_{MC}=C\frac{\sqrt{\hat \alpha_{IS}(1-\hat \alpha_{IS})}}{\sqrt{M_{MC}}\hat\alpha_{IS}},
\end{align}
and the relative error of the IS MC estimator is given by
\begin{align}\label{is}
\epsilon_{IS}=C\frac{\sqrt{\rm {var}_{p^*}\left[T^2_{\gamma_{th}}\right]}}{\sqrt{M_{IS}}\hat \alpha_{IS}},
\end{align}
where $C$ is the confidence constant equal to $1.96$ (for $95\%$ confidence interval), and $M_{MC}$ and $M_{IS}$ are the number of samples for the naive MC and the IS MC simulations, respectively. Note that the use of  $\hat \alpha_{IS}$ in (\ref{MC}) instead of $\hat \alpha_{MC}$ gives a more accurate estimate of the standard deviation of $\hat \alpha_{MC}$. 
For a fixed relative error, we define the efficiency indicator of the IS MC technique compared to the naive MC simulation as 
\begin{align}\label{k}
k=\frac{M_{MC}}{M_{IS}}=\frac{\hat \alpha_{IS}(1-\hat \alpha_{IS})}{\rm{var}_{p^*}\left [T_{th}\right ]}.
\end{align}
The more the efficiency $k$ is large, the more we need samples in the naive MC simulation to reach the relative accuracy given by IS. In other words, the bigger is $k$, the more efficient is the proposed IS technique.\par}
\subsection{Frequency of Occurrence}
{\ As it was mentioned before, a key characteristic of a good IS technique is to emphasize the sampling of important points, i.e the number of realizations satisfying $S_N \geq \gamma_{th}$. We define the frequency of occurrence as the number of samples which satisfy $S_N\geq\gamma_{th}$. In our first simulation results, we consider the sum of two i.i.d. Log-normal RVs with mean $\mu_{dB}=0\hspace{1mm}dB$ and standard deviation $\sigma_{dB}=6\hspace{1mm}dB$.
\begin{table}[t]
\begin{center}
\caption{Frequency of occurrence for the sum of two i.i.d. Log-normal with $\mu_{dB}=0 \hspace{1mm}dB$, $\sigma_{dB}=6\hspace{1mm}dB$, and $M_{IS}=M_{MC}=10^5$.} 
\label{tab1} 
\begin{tabular}{|c|c|c|c|}
  \hline 
  Threshold (dB)  & $\hat \alpha_{IS}$ & IS frequency & MC frequency \\
  \hline
  $15$ & $1.47 \times 10^{-2}$ &  $28603$  &  $1427$ \\  
  $20$ & $9.55 \times 10^{-4}$&  $27631$  &  $99$  \\
  $25$ & $3.17 \times 10^{-5}$&  $26484$  &  $3$  \\
  $30$ & $5.8 \times 10^{-7}$ & $26253$  &   $0$  \\
  $35$ &  $0.55 \times 10^{-8}$& $25982$  &   $0$  \\ 
  \hline
\end{tabular}
\end{center}
\end{table}
\begin{table}[t]
\begin{center}
\caption{Frequency of occurrence for the sum of two i.i.d. Weibull distribution  with $k=0.5$, $\beta=1$, and $M_{IS}=M_{MC}=10^5$.} 
\label{tab2} 
\begin{tabular}{|c|c|c|c|}
  \hline 
  Threshold ( dB )  & $\hat \alpha_{IS}$ & IS frequency & MC frequency \\
  \hline
  $10$ & $1.01 \times 10^{-1}$ & $29273$  &   $10097$  \\
  $15$ & $1.67 \times 10^{-2}$ &  $29270$  &  $852$ \\  
  $20$ & $1.06 \times 10^{-4}$&  $29244$  &  $6$  \\
  $25$ & $4.15 \times 10^{-8}$&  $29143$  &  $0$  \\
  $30$ &  $3.88 \times 10^{-14}$& $29049$  &   $0$  \\ 
  \hline
\end{tabular}
\end{center}
\end{table}
In Table \ref{tab1}, we have computed the frequency of occurrence using the naive MC simulation and the proposed IS technique, with $M_{MC}=M_{IS}=10^5$. Table \ref{tab1} exhibits an important feature of the IS change of measure where the frequency of realizations belonging to the rare set $S_N \geq \gamma_{th}$ is almost constant as we increase the threshold. On the other hand, the failure of sampling under the original SLN distribution is clear through its inability to construct realizations in the rare sets. 
\begin{figure}[t]
\centering
\setlength\figureheight{0.33\textwidth}
\setlength\figurewidth{0.40\textwidth}
\input{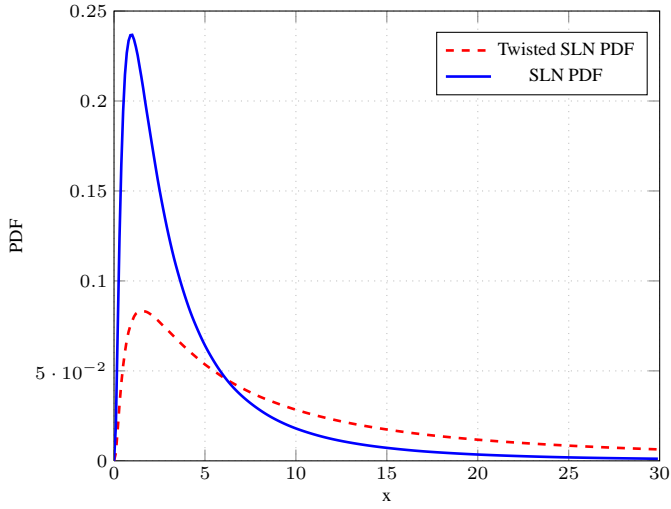}
\caption{Twisted and original PDFs of the sum of two i.i.d Log-normal RVs with $\gamma_{th}=20$, $\mu_{dB}=0\hspace{1mm}dB$, and $\sigma_{dB}=6\hspace{1mm}dB$.}
\label{twist dist fig}
\end{figure}
In Table \ref{tab2}, we show the same computation using the sum of two i.i.d Weibull distribution with shape parameter $k=0.5$ and scale parameter $\beta=1$. Again, important points are sampled more frequently using the IS technique and their frequencies remains almost constant as we increase the threshold.

To illustrate this statement, we plotted in Fig. \ref{twist dist fig} the twisted against the original SLN distributions for a fixed threshold $\gamma_{th}=20$.  Clearly, we see that twisting the hazard rate of each component in the sum leads to a more heavier twisted PDF. As a consequence, the events which exceed the given threshold are more likely to occur under the twisted PDF than under the original one.
\par}
\begin{figure}[t]
\centering
\setlength\figureheight{0.33\textwidth}
\setlength\figurewidth{0.40\textwidth}
%
%
%
%
\scalefont{0.6}
\begin{tikzpicture}

\begin{semilogyaxis}[%
width=\figurewidth,
height=\figureheight,
scale only axis,
xmin=15, xmax=34,
xlabel={$\gamma{}_{\text{th}}\text{(dB)}$},
xmajorgrids,
ymin=1e-08, ymax=0.1,
yminorticks=true,
ylabel={CCDF},
ymajorgrids,
yminorgrids,
 grid style={dotted},
legend style={at={(0.467317708333333,0.777132389187107)},anchor=south west,draw=black,fill=white,align=left}]
\addplot [
color=red,
solid,
line width=1.0pt,
mark=triangle,
mark options={solid}
]
coordinates{
 (15,0.0144631492731509)(16,0.00897544908017346)(17,0.00530196548237577)(18,0.00299982782198998)(19,0.00172595544199708)(20,0.000921603923568875)(21,0.000510716968903547)(22,0.000261776786980105)(23,0.000133980859517936)(24,6.46579902073464e-05)(25,3.31541142576449e-05)(26,1.47649808409257e-05)(27,7.03778206968063e-06)(28,3.16824271191536e-06)(29,1.31136173024463e-06)(30,5.79747499025584e-07)(31,2.51294672300438e-07)(32,9.26309963528191e-08)(33,3.6661797743031e-08)(34,1.5118996523686e-08) 
};
\addlegendentry{$\text{IS Simulation M}_{\text{IS}}\text{=5}\times\text{ 10}^{\text{4}}$};

\addplot [
color=black,
solid,
line width=1.0pt,
mark=square,
mark options={solid}
]
coordinates{
 (15,0.014798)(16,0.009006)(17,0.005159)(18,0.002989)(19,0.001673)(20,0.000942)(21,0.000481)(22,0.000243)(23,0.00013)(24,6.4e-05)(25,3.3e-05)(26,1.3e-05)(27,1e-05)(28,6e-06)(29,1e-06)(30,3e-06)(31,0) 
};
\addlegendentry{$\text{Naive Simulation M}_{\text{MC}}\text{=10}^\text{6}$};

\addplot [
color=blue,
solid,
line width=1.0pt,
mark=o,
mark options={solid}
]
coordinates{
 (15,0.01474923)(16,0.00890007)(17,0.00525938)(18,0.0030105)(19,0.0016911)(20,0.0009286)(21,0.00049909)(22,0.00025997)(23,0.00013192)(24,6.576e-05)(25,3.288e-05)(26,1.428e-05)(27,6.74e-06)(28,3.09e-06)(29,1.4e-06)(30,5.1e-07)(31,1.8e-07)(32,1.1e-07)(33,4e-08)(34,1e-08) 
};
\addlegendentry{$\text{Naive Simulation M}_{\text{MC}}\text{=10}^\text{8}$};

\end{semilogyaxis}
\end{tikzpicture}%
\caption{CCDF of the sum of two i.i.d Log-normal RVs with mean $0$ dB, and standard deviation $6$ dB.}
\label{fig1}
\end{figure}
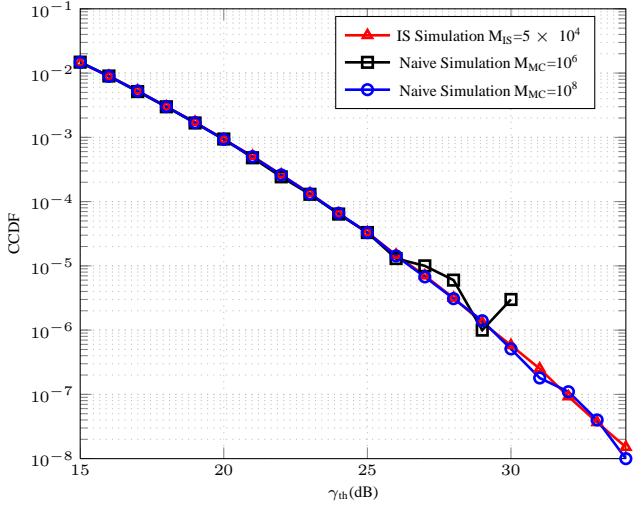
\subsection{Efficiency of the Proposed IS Algorithm}
{\ In Fig. \ref{fig1}, the CCDF of the sum of two i.i.d Log-normal RVs is presented using both the naive MC simulation and our IS simulation technique. 
The inefficiency of the naive simulation is clear in Fig. \ref{fig1}. In fact, a remarkable oscillatory behavior of the naive MC technique is observed using a number of samples $M_{MC}=10^6$ for $\gamma_{th}\geq 25$ dB. Besides, as we increase the threshold, the naive MC estimator is almost zero. Indeed, more samples are required in order to overcome this failure and to get a good approximation of the CCDF. 
The naive technique with $M_{MC}=10^8$ is also presented in Fig. \ref{fig1} and is compared to IS simulation. We point out that both methods coincide and we have a good approximation of the CCDF up to a probability of order $10^{-6}$. Then, an oscillation of the tail of the CCDF using the naive MC approach is observed, whereas IS technique gives a smooth curve. Thus, our IS technique gives a more accurate result using a less number of samples $5\times 10^4$, in contrast with $10^8$ samples used in the naive simulation. In order to confirm the previous statement, we need to analyze the relative error given by both techniques.
\par}
\begin{figure}[t]
\centering
\setlength\figureheight{0.33\textwidth}
\setlength\figurewidth{0.40\textwidth}
%
%
%
%
\scalefont{0.6}
\begin{tikzpicture}

\begin{axis}[%
width=\figurewidth,
height=\figureheight,
scale only axis,
xmin=15, xmax=34,
xlabel={$\gamma{}_{\text{th}}\text{(dB)}$},
xmajorgrids,
ymin=0, ymax=1.6,
ylabel={Relative Error},
ymajorgrids,
grid style={dotted},
legend style={at={(0.696223958333334,0.842967589797572)},anchor=south west,draw=black,fill=white,align=left}]
\addplot [
color=red,
solid,
line width=1.0pt,
mark=triangle,
mark options={solid}
]
coordinates{
 (15,0.0227745275945748)(16,0.0244973309814561)(17,0.0268346957621647)(18,0.0291475717562798)(19,0.0316757389216932)(20,0.034276897173828)(21,0.0367890646506389)(22,0.0401456190536656)(23,0.0427842592541285)(24,0.0467974271937433)(25,0.0482358040259837)(26,0.0539978861455024)(27,0.055926057140414)(28,0.0588217878676441)(29,0.0631388056767334)(30,0.0677935889144354)(31,0.0709130248993684)(32,0.0743266318024561)(33,0.0793742467807016)(34,0.0825725776493499) 
};
\addlegendentry{IS MC};

\addplot [
color=blue,
solid,
line width=1.0pt,
mark=o,
mark options={solid}
]
coordinates{
 (15,0.00161793500379201)(16,0.00205953977829415)(17,0.00268462267021703)(18,0.00357318517681155)(19,0.00471374593936907)(20,0.00645332738323871)(21,0.0086707167333998)(22,0.0121124961871274)(23,0.0169318961477096)(24,0.0243742314107281)(25,0.034039263673743)(26,0.0510077798867069)(27,0.0738816584909011)(28,0.110114918356609)(29,0.171156948068637)(30,0.257416572605642)(31,0.390988853219711)(32,0.643988161176828)(33,1.02364456479302)(34,1.5940230036231) 
};
\addlegendentry{Naive MC};

\end{axis}
\end{tikzpicture}%
\caption{Relative error of the sum of two i.i.d Log-normal RVs with mean $0$ dB, standard deviation $6$ dB, $M_{MC}=10^8$, and $M_{IS}=5\times 10^4$.}
\label{fig2}
\end{figure}
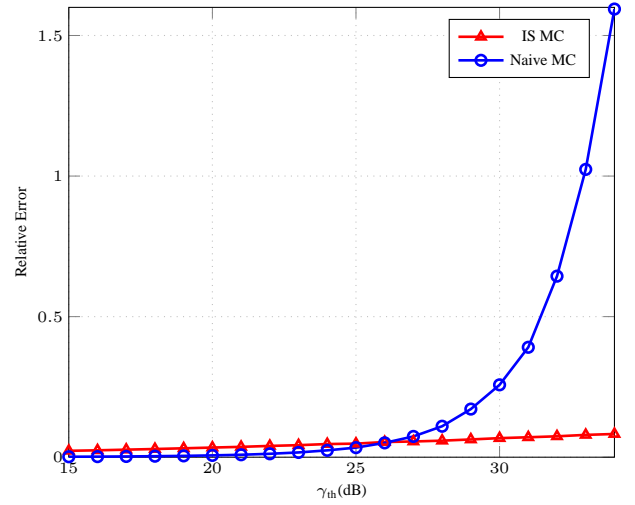
{ \ In Fig. \ref{fig2}, we plotted the relative error of the naive and the IS simulations as function of the threshold.
We point out a slow variation of the relative error of the naive MC simulation for $\gamma_{th}<25$, then a very rapid increase is observed as we increase the threshold. In fact, in the first region the number of samples is sufficient to guarantee an accurate approximation, whereas in the second region the naive simulation fails to well estimate the CCDF and hence substantial samples are required to ensure a good accuracy, i.e much more than $10^8$ realizations. On the other hand, IS technique shows an interesting result in Fig. \ref{fig2} where the variation of its relative error is extremely slow compared to the naive simulation. Consequently, with $M_{IS}$ much smaller than $M_{MC}$, our IS approach approximates the CCDF more efficiently than the naive simulation.
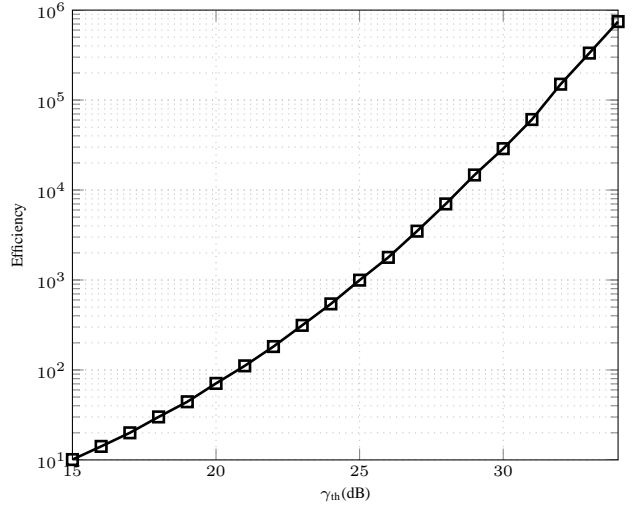
\begin{figure}[t]
\centering
\setlength\figureheight{0.33\textwidth}
\setlength\figurewidth{0.40\textwidth}
%
%
%
%
\scalefont{0.6}
\begin{tikzpicture}

\begin{semilogyaxis}[%
width=\figurewidth,
height=\figureheight,
scale only axis,
xmin=15, xmax=34,
xlabel={$\gamma{}_{\text{th}}\text{(dB)}$},
xmajorgrids,
ymin=10, ymax=1000000,
yminorticks=true,
ymajorgrids,
yminorgrids,
grid style={dotted},
ylabel={Efficiency}]
\addplot [
color=black,
solid,
line width=1.0pt,
mark=square,
mark options={solid},
forget plot
]
coordinates{
 (15,10.0937695017197)(16,14.1362046745169)(17,20.0171917990732)(18,30.0563530026784)(19,44.2903210845611)(20,70.8915362806612)(21,111.097043058076)(22,182.062702045618)(23,313.237070262072)(24,542.560160795057)(25,995.981486179885)(26,1784.63464863811)(27,3490.39729400764)(28,7008.83565596515)(29,14696.9202738764)(30,28835.3848218922)(31,60800.4404716958)(32,150140.154125769)(33,332635.909274738)(34,745328.074899772) 
};
\end{semilogyaxis}
\end{tikzpicture}%
\caption{Efficiency of the sum of two i.i.d Log-normal RVs with mean $0$ dB, standard deviation $6$ dB, $M_{MC}=10^8$, and $M_{IS}=5\times 10^4$.}
\label{fig3}
\end{figure}
\par}
{\ In Fig. \ref{fig3}, we plotted the efficiency indicator $k$ as function of the threshold. From this figure, we deduce that the efficiency is increasing rapidly, almost exponentially. Hence, the more we increase the threshold the more efficient is our IS technique. This result is expected since $k$ is proportional to the number of samples $M_{MC}$ that we need to generate in order to absorb the rapid increase of the relative error of the naive MC simulation, i.e reach the relative accuracy given by the IS approach. Besides, Fig. \ref{fig3} illustrates also that the IS technique is more efficient for the considered range of probability, i.e $k$ always bigger than 1.  
\par}   
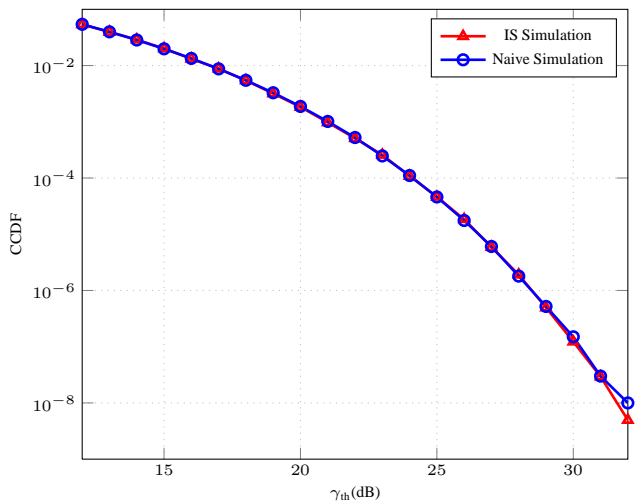
\begin{figure}[t]
\centering
\setlength\figureheight{0.33\textwidth}
\setlength\figurewidth{0.40\textwidth}
%
%
%
%
\scalefont{0.6}
\begin{tikzpicture}

\begin{semilogyaxis}[%
width=\figurewidth,
height=\figureheight,
scale only axis,
xmin=12, xmax=32,
xlabel={$\gamma{}_{\text{th}}\text{(dB)}$},
xmajorgrids,
ymin=1e-09, ymax=0.1,
yminorticks=true,
ylabel={CCDF},
ymajorgrids,
yminorgrids,
 grid style={dotted},
legend style={draw=black,fill=white,align=left}]
\addplot [
color=red,
solid,
line width=1.0pt,
mark=triangle,
mark options={solid}
]
coordinates{
 (12,0.054016412360405)(13,0.0396945092118784)(14,0.0289023548159065)(15,0.0200772571448005)(16,0.0131873778230559)(17,0.00880247544094897)(18,0.005440435943709)(19,0.00317852875498783)(20,0.00183586347701705)(21,0.000974508315187449)(22,0.00050896916974626)(23,0.000257413508855184)(24,0.000108621893261887)(25,4.62073666677634e-05)(26,1.84858692154886e-05)(27,5.90664433618488e-06)(28,1.90974128415054e-06)(29,4.95100603157127e-07)(30,1.22761616110306e-07)(31,2.93442342016019e-08)(32,4.96039066487983e-09) 
};
\addlegendentry{IS Simulation};

\addplot [
color=blue,
solid,
line width=1.0pt,
mark=o,
mark options={solid}
]
coordinates{
 (12,0.05416289)(13,0.03977758)(14,0.02849642)(15,0.01985949)(16,0.01340922)(17,0.00873453)(18,0.00546904)(19,0.00328279)(20,0.00187407)(21,0.00101443)(22,0.00052514)(23,0.00024704)(24,0.00011072)(25,4.599e-05)(26,1.758e-05)(27,6.06e-06)(28,1.8e-06)(29,5.2e-07)(30,1.5e-07)(31,3e-08)(32,1e-08) 
};
\addlegendentry{Naive Simulation};

\end{semilogyaxis}
\end{tikzpicture}%
\caption{CCDF of the sum of two independent Weibull RVs with $\beta_1=\beta_2=1$, $k_1=0.4$, $k_2=0.8$, $M_{IS}=5 \times 10^4$, and $M_{MC}=10^8$.}
\label{fig4}
\end{figure}
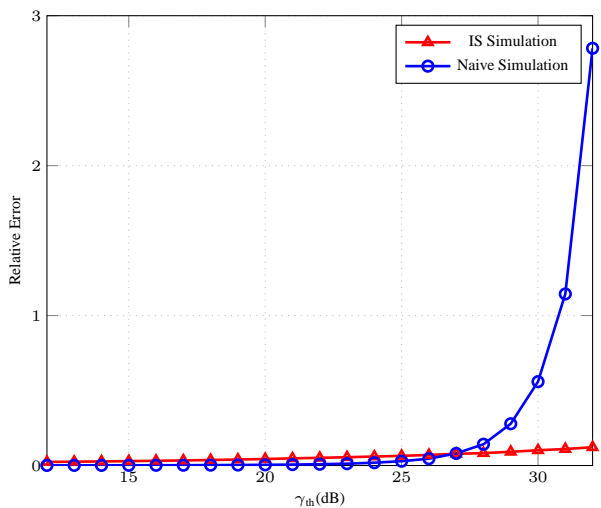
\begin{figure}[t]
\centering
\setlength\figureheight{0.33\textwidth}
\setlength\figurewidth{0.40\textwidth}
%
%
%
%
\scalefont{0.6}
\begin{tikzpicture}

\begin{axis}[%
width=\figurewidth,
height=\figureheight,
scale only axis,
xmin=12, xmax=32,
xlabel={$\gamma{}_{\text{th}}\text{(dB)}$},
xmajorgrids,
ymin=0, ymax=3,
ylabel={Relative Error},
ymajorgrids,
grid style={dotted},
legend style={draw=black,fill=white,align=left}]
\addplot [
color=red,
solid,
line width=1.0pt,
mark=triangle,
mark options={solid}
]
coordinates{
 (12,0.0238971007385092)(13,0.0254932924213162)(14,0.0272486912487197)(15,0.0291601750640427)(16,0.0317158096601625)(17,0.033826145963268)(18,0.036795408655573)(19,0.0400326811075467)(20,0.0434935933063397)(21,0.047760644052455)(22,0.0513201408498082)(23,0.055336399583328)(24,0.059989050200633)(25,0.0642881695422758)(26,0.0705516725086687)(27,0.0781289934777545)(28,0.0826598221258691)(29,0.0916241449510802)(30,0.101693785413103)(31,0.109979298476464)(32,0.122042673610939) 
};
\addlegendentry{IS Simulation};

\addplot [
color=blue,
solid,
line width=1.0pt,
mark=o,
mark options={solid}
]
coordinates{
 (12,0.000820228781012139)(13,0.000964041121034194)(14,0.00113611137394775)(15,0.00136930375558535)(16,0.00169548594803438)(17,0.00207985838170784)(18,0.00265005336637392)(19,0.00347097664136367)(20,0.00457021741930358)(21,0.00627554693176206)(22,0.00868559918085474)(23,0.0122147476813604)(24,0.018805019973063)(25,0.0288330395581599)(26,0.0455860664351802)(27,0.0806462872300733)(28,0.141830093221011)(29,0.278553895380329)(30,0.559402942084033)(31,1.14418085033969)(32,2.78290338804418) 
};
\addlegendentry{Naive Simulation};

\end{axis}
\end{tikzpicture}%
\caption{Relative error of the sum of two independent Weibull RVs with $\beta_1=\beta_2=1$, $k_1=0.4$, $k_2=0.8$, $M_{IS}=5 \times 10^4$, and $M_{MC}=10^8$.}
\label{fig5}
\end{figure}
{\ In the second simulation results, we consider the sum of two independent Weibull distribution with same scale parameter $\beta=1$, and with different shape parameters $k_1=0.4$, and $k_2=0.8$. In Fig. \ref{fig4}, Fig. \ref{fig5}, and Fig. \ref{fig6}, we plotted the CCDF, the relative error, and the efficiency, respectively. We note that in this case also, the proposed IS technique gives an accurate and efficient approximation of the CCDF and results in a substantial computational gain.
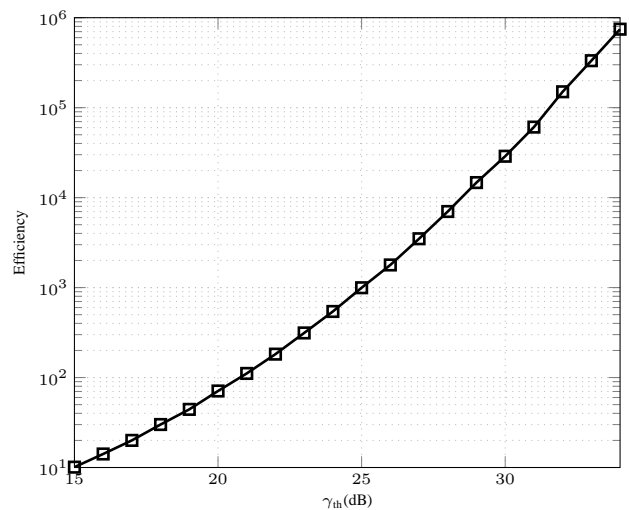
\begin{figure}[t]
\centering
\setlength\figureheight{0.33\textwidth}
\setlength\figurewidth{0.40\textwidth}
%
%
%
%
\scalefont{0.6}
\begin{tikzpicture}

\begin{semilogyaxis}[%
width=\figurewidth,
height=\figureheight,
scale only axis,
xmin=15, xmax=34,
xlabel={$\gamma{}_{\text{th}}\text{(dB)}$},
xmajorgrids,
ymin=10, ymax=1000000,
yminorticks=true,
ymajorgrids,
yminorgrids,
grid style={dotted},
ylabel={Efficiency}]
\addplot [
color=black,
solid,
line width=1.0pt,
mark=square,
mark options={solid},
forget plot
]
coordinates{
 (15,10.0937695017197)(16,14.1362046745169)(17,20.0171917990732)(18,30.0563530026784)(19,44.2903210845611)(20,70.8915362806612)(21,111.097043058076)(22,182.062702045618)(23,313.237070262072)(24,542.560160795057)(25,995.981486179885)(26,1784.63464863811)(27,3490.39729400764)(28,7008.83565596515)(29,14696.9202738764)(30,28835.3848218922)(31,60800.4404716958)(32,150140.154125769)(33,332635.909274738)(34,745328.074899772) 
};
\end{semilogyaxis}
\end{tikzpicture}%
\caption{Efficiency of the sum of two independent Weibull RVs with $\beta_1=\beta_2=1$, $k_1=0.4$, $k_2=0.8$, $M_{IS}=5 \times 10^4$, and $M_{MC}=10^8$.}
\label{fig6}
\end{figure}
\par}
\subsection{Near-Optimality of the Minmax Twisting Parameter}
In our next simulation results, we aim to analyze the sensibility of the second moment of the RV $T_{\gamma_{th}}$ with respect to the twisting parameter $\theta$. Since our twisting parameter $\theta^*$ given in (\ref{theta}) is chosen to minimize an upper bound on $\mathbb{E}_{\theta} \left [ T_{\gamma_{th}}^2\right ]$, we need to investigate whether $\theta^*$ is close to the optimal unknown twisting parameter, that is the value that minimizes the actual value of $\mathbb{E}_{\theta} \left [ T_{\gamma_{th}}^2 \right ]$. We consider the sum of two i.i.d Weibull RVs with shape and scale parameters equal to $0.5$ and $1$,  respectively. In Fig. \ref{fig7}, we plot the upper bound (\ref{bound}) and the actual value of $\mathbb{E}_{\theta} \left [ T_{\gamma_{th}}^2\right ]$  function in $\theta$ and for different value of the threshold $\gamma_{th}$. We note that the exact computation of $\mathbb{E}\left [ T_{\gamma_{th}}^2\right ]$ has a unique minimum which is closer to our choice $\theta^*$. Moreover, as $\gamma_{th}$ increases, the difference between the two minimizers becomes negligible and thus we tend to the optimal value. Another important deduction is that the second moment is slowly varying with respect to the twisting parameter $\theta$ especially in the neighborhood of the optimal value. Hence, our choice of $\theta^*$ is actually reasonable since it almost results in approximately the biggest reduction of variance.
\begin{figure}[t]
\centering
\setlength\figureheight{0.17\textwidth}
\setlength\figurewidth{0.19\textwidth}
\input{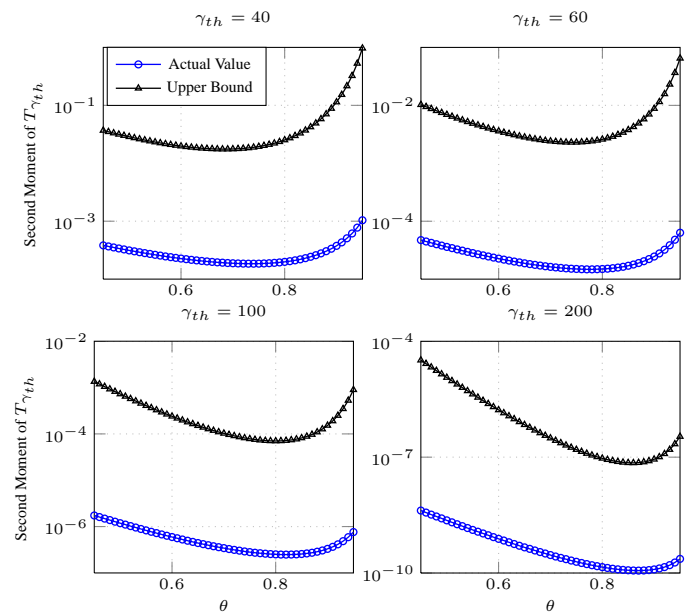}
\caption{Actual value and upper bound of $\mathbb{E} \left [ T_{\gamma_{th}}^2 \right ]$ function in $\theta$ for the sum of two i.i.d Weibull RVs with $k_1=k_2=0.5$ ,and $\beta=1$.}
\label{fig7}
\end{figure}
\par}
\section{Conclusion}
In this paper, we developed an efficient hazard rate twisting  technique for the estimation of the probability that a sum of independent RVs exceeds any given threshold.  We presented a general procedure to find the best possible twisting parameter which leads to the possible largest reduction of the IS estimator variance for all possible values of the threshold. Besides, this approach, which seems to be consistent with the class of subexponential distributions, results in ensuring the asymptotic optimality criterion as the threshold goes to infinity. Numerical simulations showed that the optimized IS approach could reach the same accuracy as the naive MC simulation with a substantial computational gain. This alternative technique could serve as a benchmark to study the accuracy of  future closed-form approximations of the quantity of interest. 
\bibliography{biblio}
\bibliographystyle{IEEEtran}
\end{document}